\documentclass[conference]{IEEEtran}
\IEEEoverridecommandlockouts
\usepackage{cite}
\usepackage{algorithmic}
\usepackage{graphicx}
\usepackage{textcomp}
\usepackage{xcolor}
\def\BibTeX{{\rm B\kern-.05em{\sc i\kern-.025em b}\kern-.08em
    T\kern-.1667em\lower.7ex\hbox{E}\kern-.125emX}}
\graphicspath{{figures/}}

\newcommand{\polo}{}
\newcommand{\louis}{}

\usepackage[colorlinks=true, allcolors=blue]{hyperref}

\usepackage{amsmath,amsfonts,stmaryrd,dsfont}
\usepackage{multirow}
\usepackage{bm}

\def\x{{\mathbf x}}

\def\C{\mathbb{C}}
\def\R{\mathbb{R}}
\newcommand{\vct}[1]{\bm{#1}}
\newcommand{\mtx}[1]{\bm{#1}}
\newcommand{\nat}[1]{{#1}^{\natural}}
\newcommand{\norm}[1]{\left\lVert#1\right\rVert}
\newcommand{\abs}[1]{\left\lvert#1\right\rvert}

\DeclareMathOperator*{\diag}{diag}
\DeclareMathOperator{\Tr}{trace}
\def\xvbar{\vct{x}_{\bar{v}}}
\def\Phivbar{\mtx{\Phi}_{\bar{v}}}
\def\xvnat{\nat{\vct{x}}_v}
\def\H{\mathsf{H}}

\usepackage{algorithm}
\usepackage{algorithmic}

\usepackage{array}

\usepackage{url,cite}
\usepackage{xurl} 

\begin{document}

\title{Signal Inpainting from Fourier Magnitudes\thanks{This work was made with the support of the French National Research Agency through project DENISE (ANR-20-CE48-0013), and was conducted while L. Bahrman was an intern with Inria in Nancy, France.}}

\author{
\IEEEauthorblockN{Louis Bahrman\textsuperscript{1,2}}
\IEEEauthorblockA{\textsuperscript{1}\textit{LTCI, Télécom Paris, Institut Polytechnique de Paris} \\
Paris, France \\
louis.bahrman@telecom-paris.fr}
\and
\IEEEauthorblockN{Marina Krémé\textsuperscript{2}, Paul Magron\textsuperscript{2}, Antoine Deleforge\textsuperscript{2}}
\IEEEauthorblockA{\textsuperscript{2}\textit{Université de Lorraine, CNRS, Inria, LORIA} \\
Nancy, France \\
\{ama-marina.kreme, paul.magron, antoine.deleforge\}@inria.fr}
}

\maketitle

\begin{abstract}
Signal inpainting is the task of restoring degraded or missing samples in a signal. In this paper we address signal inpainting when discrete Fourier magnitudes are observed. We propose a mathematical formulation of the problem that highlights its connection with phase retrieval, and we introduce two methods for solving it. First, we derive an alternating minimization scheme, which shares similarities with the Gerchberg-Saxton algorithm, a classical phase retrieval method. Second, we propose a convex relaxation of the problem, which is inspired by recent approaches that reformulate phase retrieval into a semi-definite program. We assess the potential of these methods for the task of inpainting gaps in speech signals. Our methods exhibit both a high probability of recovering the original signals and robustness to magnitude noise.
\end{abstract}


\IEEEpeerreviewmaketitle

\section{Introduction}
\label{sec:intro}

Signal inpainting \cite{adler2012inpainting} is an inverse problem that consists in restoring signals degraded by sample loss. Such a problem typically arises as a result of degradation during signal transmission (packet loss concealment~\cite{Rodbro2006}) or in digitization of physically degraded media.
Inpainting can also be used to restore signal samples subject to a degradation so heavy that the information about the samples can be considered lost (e.g., signal clipping~\cite{Chantas2018} or impulsive noises~\cite{Derebssa2018}). More specifically, in this paper we focus on inpainting compact gaps, which occurs, e.g., when an audio signal is corrupted with clicks~\cite{magron2015linearunwrapping}.

Approaches that tackle this issue can be divided into two categories depending on the number of missing  samples or duration of the gaps. When considering short gaps (less than $100$ ms), approaches based on autoregressive modeling~\cite{janssen1986interpautoregressive}, convex optimization~\cite{Taubock2021}, sparse modeling~\cite{Mokry2019}, or Bayesian estimation \cite{godsill1995bayesianrestauration} have shown promising results. Conversely, approaches based on sinusoidal modeling \cite{lagrange2005long} or graphs \cite{perraudin2018inpaintinggraphs} are more suitable for longer gaps (more than $100$ ms). However, signal inpainting remains a challenging problem, and these approaches are not adapted to scenarios where some additional information about the signal is available.

We focus here on a setting where the missing gap is contained in a short context window, while the \textit{discrete Fourier magnitudes} of the complete signal on that segment are observed. A similar magnitude-informed setting was notably studied in~\cite{Deleforge2017} and~\cite{liutkus2018audio} in the different context of source separation. Beyond the fact that \louis{this} problem is open and \louis{has} not received specific attention in the literature, the motivation behind studying it is the existence of a vast literature dedicated to modeling and processing short-term Fourier magnitudes in the audio literature, e.g., using nonnegative matrix factorization~\cite{Fevotte2009} or more recently variational auto-encoders~\cite{girin2019notes}. This is because the Fourier magnitudes of natural signals tend to exhibit smoother and hence more predictable evolution than their respective Fourier phases. 
\polo{
In this work, we assume that such a magnitude model has been leveraged beforehand, and we focus on the latest part of the inpainting problem, thus considering that some magnitude estimates are available, possibly up to errors.
}

The problem then shares a close connection with phase retrieval~\cite{walther63}, the task of retrieving a signal from nonnegative measurements (usually magnitudes of a set of inner products). From the seminal early works of Gerchberg and Saxton~\cite{Gerchberg1972}, this task has been revived over the last decade with the development of novel optimization approaches based on gradient descent~\cite{WF} or convex relaxations~\cite{candes2013PhaseLift, waldspurger2015phase}. While phase retrieval has experienced considerable progress in recent years, its connection to signal inpainting is largely left to explore, from the theoretical, methodological and applicative standpoints.

This papers aims at contributing to bridging this gap by formulating signal inpainting from Fourier magnitudes as a constrained phase retrieval problem. Inspired by phase retrieval algorithms, we derive two methods to tackle it, based on alternating minimization (AM) and convex relaxation (CR). Experiments conducted on speech signals reveal the potential of these techniques for audio inpainting, since they exhibit a high probability of recovering the original signal, as well as some robustness to magnitude errors. The approach compares favorably to the sate-of-the-art sparsity-based inpainting method in \cite{Mokry2019} as long as magnitudes are observed with sufficient accuracy.

The rest of this paper is organized as follows. Section \ref{sec:methods} formulates the problem and introduces two methods for solving it. Experiments on speech signals are conducted in Section \ref{sec:experiments}. Finally, Section \ref{sec:conclusion} draws some concluding remarks.

\section{Methods}
\label{sec:methods}

\subsection{From inpainting to phase retrieval}
\label{sec:inpaintinToPR}

\begin{figure}[t]
    \centering
    \includegraphics[width=1\linewidth]{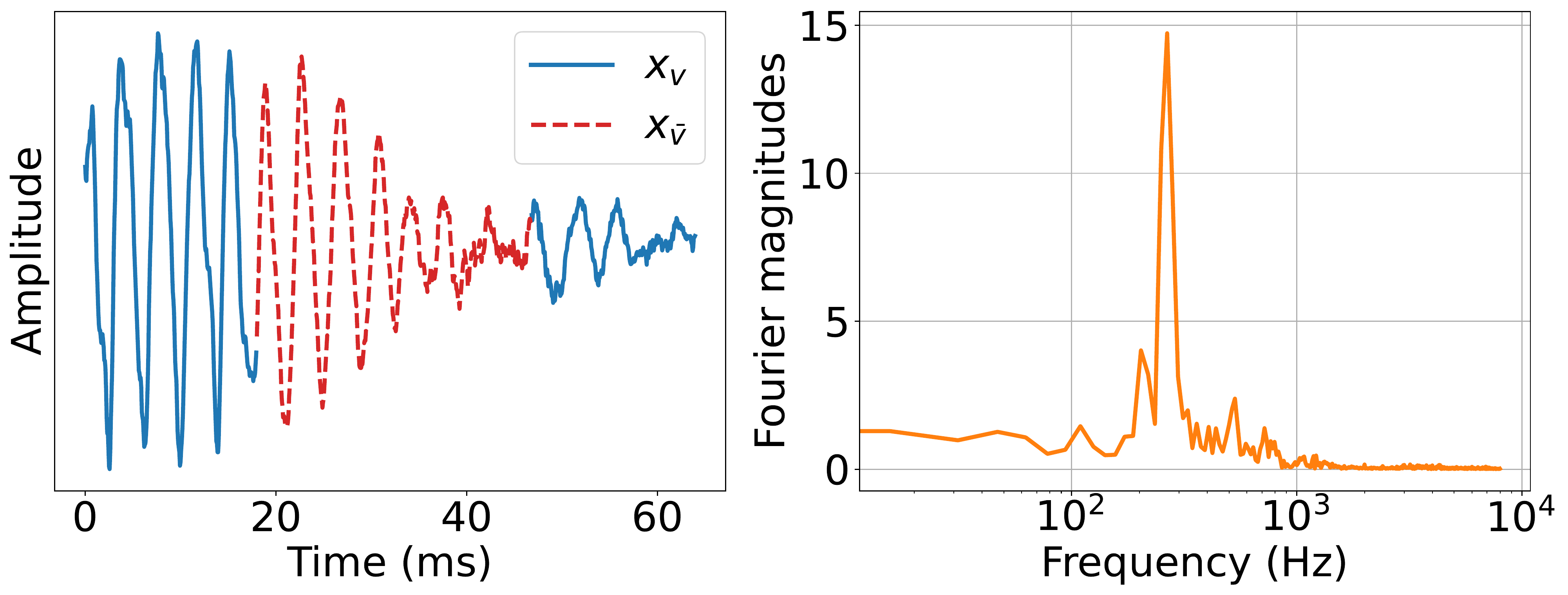}
    \vspace{-1.2mm}
    \caption{Illustration of restoring the missing samples of a signal (left), assuming its Fourier magnitudes (right) \louis{have been estimated beforehand}.\vspace{-5mm}}
    \label{fig:inpainting}
\end{figure}

Let $\vct{x}^{\natural} \in \R^L$ denote a signal. We partition its support $\{0,\dots,L-1\}$ into two sets $v$ and $\bar{v}$ such that $\xvnat \in \R^{L-d}$ and $\xvbar^{\natural} \in \R^{d}$ denote the observed and missing samples, respectively, and where $d \leq L$ denotes the number of missing samples, whose location $\bar{v}$ is assumed to be known. 
\louis{Note that the missing samples $\x_{\bar{v}}$ can be non-contiguous.}
We also assume that the magnitudes of its discrete Fourier transform (DFT) $\vct{b} \in \R_+^L$ are observed. We consider a \textit{complete} DFT (that is, it does not use zero-padding), whose matrix is denoted $\mtx{\Phi} \in \C^{L \times L} $.
The task of signal inpainting from Fourier magnitudes is illustrated in Fig.~\ref{fig:inpainting} and can be stated as:
\begin{equation}
\label{eq:formulationInpainting}
    \mathrm{Find} \quad {\vct{x} \in \R^{L}} \quad \text{such that}\;  \abs{\mtx{\Phi}  \vct{x}}  =  \vct{b} \; \text{and}\; \vct{x}_{v} = \xvnat.
\end{equation}
Let us formulate~\eqref{eq:formulationInpainting} as the following optimization problem:
\begin{equation}
\label{eq:inpainting_optim_problem}
	\min_{\vct{x} \in \R^{L}} \norm{ | \mtx{\Phi} \vct{x} | - \vct{b} }^2 \quad \text{s. t.} \quad \vct{x}_{v}=\xvnat,
\end{equation}
where $\norm{.}$ denotes the Euclidean norm\footnote{Recent works such as~\cite{Li2021, Vial2021} have investigated alternative discrepancy measures for phase retrieval. We focus on the Euclidean norm in this study.}. Problem~\eqref{eq:inpainting_optim_problem} reads as a constrained phase retrieval problem. If the whole signal is missing (${d=L}$), the constraint vanishes and it reduces to Fourier phase retrieval. We theoretically show in a supplementary material that for almost all $\vct{x}^{\natural}\in \R^L$, this problem admits a \textit{unique solution} if at most $33\%$ of the signal is missing \cite{kreme2023inpainting}.
Note that Fourier phase retrieval has been extensively investigated from a theoretical perspective; we refer the interested reader to the book chapter~\cite{Bendory2017} for a comprehensive review on uniqueness guarantees. The cases where one sample or half of the samples at the end of the signal are observed were respectively covered some time ago in~\cite{xu1987almost} and in~\cite{Nawab1983}. However, both studies considered \louis{DFT with zero-padding.}
Our supplemental theorem complements these results by treating the more challenging 
\louis{unpadded} DFT case, using a different proof technique.

Notwithstanding our uniqueness result, the problem is non-convex and can be viewed as an instance of quadratic programming, which is known to be NP-hard in general \cite{audet2000branch}. As pointed out in the review \cite{Bendory2017}, \textit{``there is }[currently]\textit{ no algorithm that knows how to exploit the given entries to recover the complete signal in a stable and efficient manner"}. This highlights the need for efficient algorithms that provide a solution. The cornerstone of our approaches lies in introducing an auxiliary phase variable $\vct{u} \in \C^L$ such that $\abs{\vct{u}}=1$. Then, ${\vct{b} = |\diag(\vct{b}) \vct{u}|}$, where $\diag(\vct{b})$ is the diagonal matrix whose entries are given by the vector $\vct{b}$. We then turn our attention to minimizing the following auxiliary function, which is \textit{exact} with respect to \eqref{eq:inpainting_optim_problem} in the sense of \cite{yevtushenko1990exact}:
\begin{equation}
	\min_{\vct{x} \in \R^{L}, \vct{u} \in \C^{L}} \norm{  \mtx{\Phi} \vct{x} - \diag(\vct{b}) \vct{u} }^2 \; \text{s.t.} \; \vct{x}_{v}=\xvnat,  \; \abs{\vct{u}} = 1.
	\label{eq:inpainting_prob_formulation_ux}
\end{equation}

\subsection{Alternating minimization}
The first approach we propose to solve~\eqref{eq:inpainting_prob_formulation_ux} is an AM scheme. Let us first fix $\vct{u}$ and derive the update for $\vct{x}$, for which we propose to incorporate the constraint $\vct{x}_{v}=\xvnat$ within the loss function. To that end, let us reorder $\vct{x}$ as ${\vct{x} = \begin{bmatrix}
\vct{x}_{\bar{v}} \\
\vct{x}_{v}
\end{bmatrix}}$ (and similarly for $\vct{x}^{\natural}$), and split $\mtx{\Phi}$ accordingly as $\mtx{\Phi} = [\mtx{\Phi}_{\bar{v}}, \mtx{\Phi}_{v}]$, with
$\mtx{\Phi}_{\bar{v}} \in \C^{L \times d}$ and $\mtx{\Phi}_{v} \in \C^{L \times (L-d)}$. Using these notations, we have ${\mtx{\Phi} \vct{x} = \mtx{\Phi}_{\bar{v}}\vct{x}_{\bar{v}}+ \mtx{\Phi}_{v}\vct{x}_v}$, and~\eqref{eq:inpainting_prob_formulation_ux} rewrites:
\begin{equation}
	\min_{\vct{x}_{\bar{v}} \in \R^{d}} \norm{  \mtx{\Phi}_{\bar{v}}\vct{x}_{\bar{v}}+ \mtx{\Phi}_{v}\vct{x}_v^{\natural} -  \diag(\vct{b}) \vct{u}}^2.
	\label{eq:AM_prob_x}
\end{equation}
Since $\mtx{\Phi}_{\bar{v}}$ is full-rank, it has a left inverse which is its Hermitian transpose $\mtx{\Phi}_{\bar{v}}^\H$. Besides, recall that since the DFT is an orthogonal transform, then $\Phivbar^{\H} \mtx\Phi_v = 0$. Altogether this yields the following solution to~\eqref{eq:AM_prob_x}:
\begin{equation}
	    \xvbar  = \mtx{\Phi}_{\bar{v}}^{\H}\diag(\vct{b})\vct{u}.
    \label{eq:AM_x}
\end{equation}
Let us now derive the update for $\vct{u}$ when $\vct{x}$ is fixed, for which~\eqref{eq:inpainting_prob_formulation_ux} rewrites:
\begin{equation}
	\min_{\vct{u} \in \C^{L}} \norm{  \mtx{\Phi} \vct{x} - \diag(\vct{b}) \vct{u} }^2 \quad \text{s. t.} \quad \abs{\vct{u}} = 1,
	\label{eq:AM_prob_u}
\end{equation}
which is straightforward to solve:
\begin{equation}
	    \vct{u} = \frac{\mtx{\Phi} \vct{x}}{\abs{ \mtx{\Phi} \vct{x}} }.
    \label{eqn:AM_solution_u}
\end{equation}
Alternating~\eqref{eq:AM_x} and~\eqref{eqn:AM_solution_u} yields a procedure summarized in Algorithm~\ref{al:AlternatedMinimization}. We discuss the initialization strategy in Section~\ref{sec:exp_setting}. Note that at line~\ref{al_line:projectrealAxes} we apply the real part function to ensure a real-valued signal estimate\footnote{It can be proven rigorously that doing so does not affect convergence guarantees as in~\cite{Vial2021}, but we do not detail this here for brevity.}.

\noindent \textit{Remark}: Algorithm~\ref{al:AlternatedMinimization} consists in computing the DFT of a signal, setting its magnitude to a target value, inverting the DFT, putting back the observed samples, and repeating. As such, it is similar to the Gerchberg-Saxton algorithm~\cite{Gerchberg1972} except that the signal-domain magnitude constraint is here replaced with a projection onto the partially observed samples.
Besides, Algorithm~\ref{al:AlternatedMinimization} can be seen as a member of the general family of \textit{constrained phase retrieval} algorithms reviewed in~\cite{Bendory2017}. While such algorithms are generally heuristically derived, the proposed alternate minimization scheme on an exact auxiliary function is guaranteed to converge - though not necessarily to a global minimum - by construction.

\begin{algorithm}[t]
	\caption{AM for signal inpainting}
	\label{al:AlternatedMinimization}
	\begin{algorithmic}[1]
		\REQUIRE
		$\begin{cases}
			\vct{b} \in \R^{L}_{+}:  \text{observations}, \; \nat{\vct{x}}_{v}: \text{known signal}\\
		    \mtx{\Phi}: \text{Fourier matrix}, \\
		\end{cases}$
		\STATE Initialize $\vct{x}^{(0)}_{\bar{v}}$  \quad \text{and} \quad $\vct{x}^{(0)} \gets \begin{bmatrix}
		 \vct{x}^{(0)}_{\bar{v}}\\ \vct{x}^{\natural}_{v}  \end{bmatrix}$
		 \label{initialisation}
		\STATE $i \gets 0$
		\WHILE{convergence not reached}
		\STATE  $\vct{u}^{(i+1)} \gets \frac{\mtx{\Phi} \vct{x}^{(i)}}{\abs{ \mtx{\Phi} \vct{x}^{(i)} }}$
		\label{estimatePhases}
		\STATE $\xvbar^{(i+1)} \gets \Re\left( \mtx{\Phi}_{\bar{v}}^{\H}\diag(\vct{b}) \vct{u}^{(i+1)}\right) $
		\label{al_line:projectrealAxes}
		\STATE $\vct{x}^{(i+1)} \gets \begin{bmatrix}
		\xvbar^{(i+1)} \\ \xvnat 
		\end{bmatrix}$
		 \label{signalDomainConstraint}
		\STATE $i \gets i+1$
		\ENDWHILE
	\ENSURE{Reconstructed signal $\vct{x}^{(i)}$}
	\end{algorithmic}
\end{algorithm}

\subsection{Convex relaxation}
 
Let us now derive a method inspired by the PhaseCut algorithm~\cite{waldspurger2015phase}, which is based on a CR of the problem. This approach consists in reformulating phase retrieval as a constrained trace minimization problem by lifting it to a higher dimensional space and relaxing the rank-one constraint. The resulting problem can then be efficiently solved via semi-definite programming. We consider the formulation~\eqref{eq:inpainting_prob_formulation_ux} in which we inject the expression of $\vct{x}$ given by~\eqref{eq:AM_x}:
\begin{equation}
	\min_{\vct{u} \in \C^{L}}{\norm{(\mtx{\Phi}_{\bar{v}}\mtx{\Phi}_{\bar{v}}^{\H}\! - \!\mtx{I})\diag(\vct{b})\vct{u} \!+ \!\mtx{\Phi}_{v}\vct{x}_v}^2}
	\; \text{s.t.} \; \abs{\vct{u}} \!= \!1.
	\label{eq:affinePhasecut}
\end{equation}
Now, let us introduce the following:
\begin{equation}
\widetilde{\vct{m}} := [(\mtx{\Phi}_{\bar{v}}\mtx{\Phi}_{\bar{v}}^{\H}\! - \!\mtx{I})\diag(\vct{b}),\mtx{\Phi}_{v}\vct{x}_v ]\; \text{and}\; \widetilde{\vct{u}} = \begin{bmatrix}
 \vct{u}\\
 1
 \end{bmatrix},
\label{eqn:mtilde}
\end{equation}
from which we can rewrite \eqref{eq:affinePhasecut} as:
\begin{equation}
	\min_{\widetilde{\vct{u}}\in \C^{L+1}}{\norm{\widetilde{\vct{m}} \widetilde{\vct{u}} } ^2 \quad \text{s.t.} \quad |\widetilde{\vct{u}}|= 1 \;  \text{and} \; \widetilde{\vct{u}}[L]=1.}
	\label{eqn:phaseMinimisation}
\end{equation}
Drawing on~\cite{waldspurger2015phase}, we lift and relax \eqref{eqn:phaseMinimisation} to the following convex problem:
\begin{equation}
\min_{\widetilde{\vct{U}}\in \C^{(L+1)\times (L+1)}}{\!\Tr{(\widetilde{\vct{M}}\widetilde{\vct{U}})\!} \quad \text{s.t.} \quad \diag{(\widetilde{\vct{U}})}=\boldsymbol{1}, \; \widetilde{\vct{U}} \succeq 0},
	\label{eqn:traceMinimisationRelaexed}
\end{equation}
where $\widetilde{\vct{M}} = \widetilde{\vct{m}}^{\H}\widetilde{\vct{m}} \in \C^{(L+1)\times (L+1)}$  and $\widetilde{\vct{U}} = \widetilde{\vct{u}}\widetilde{\vct{u}}^{\H}$. 
As in~\cite{waldspurger2015phase} and \cite{Deleforge2017}, problem~\eqref{eqn:traceMinimisationRelaexed} can be solved by means of a block coordinate descent algorithm, which we summarize in Algorithm~\ref{al:bcd} (lines \ref{al_line:bcd_deb} to \ref{al_line:bcd_fin}). It consists of a nested loop where at each iteration $i$, the columns of $\mtx{\widetilde{U}}$ are updated sequentially using the notation:
\begin{equation}
    k^{c} = \{ 0,\ldots,k-1, k+1, \ldots, L-1\}.
\end{equation}
This yields a global solution $\mtx{\widetilde{U}}$ to problem \eqref{eqn:traceMinimisationRelaexed}. If this solution is of rank 1, i.e., $\widetilde{\vct{U}} = \widetilde{\vct{u}}\widetilde{\vct{u}}^{\H}$, then $\widetilde{\vct{u}}[0,\dots,L-1]/\widetilde{\vct{u}}[L]$ is guaranteed to globally solve \eqref{eq:affinePhasecut}. We then obtain a global solution of \eqref{eq:inpainting_prob_formulation_ux} via:
\begin{equation}
    \xvbar =  \Phivbar^{\H}\diag{(\vct{b})} \widetilde{\vct{u}}[0,\ldots,L-1]/\widetilde{\vct{u}}[L].
\end{equation}
However, $\mtx{\widetilde{U}}$ needs not be rank-1 in general. Hence, as commonly done in semi-definite relaxations, its closest rank-1 approximation is used in practice, which by the Eckart-Young-Mirsky theorem amounts to setting $\widetilde{\vct{u}}$ to the eigenvector associated to the largest eigenvalue of $ \mtx{\widetilde{U}}$. 
\begin{algorithm}[t]
	\caption{CR for signal inpainting}
	\label{al:bcd}
	\begin{algorithmic}[1]
		\REQUIRE
		$\begin{cases}
			\vct{b} \in \R^{L}_{+}:  \text{observations},\; \xvnat: \text{known signal} \\
			 \mtx{\Phi}: \text{Fourier matrix},\; \nu \geq 0: \text{barrier parameter}\\
			\widetilde{\vct{m}} \in \C^{L\times (L+1)}: \text{matrix defined by \eqref{eqn:mtilde}}\\
		    n_\text{iter}: \text{number of iterations}
		    \end{cases}$
		\STATE $ \widetilde{\vct{M}} \gets \widetilde{\vct{m}}^{\H}\widetilde{\vct{m}}$ \quad and \quad $\mtx{\widetilde{U}}^{(0)} \gets \mtx{I}$ \label{al_line:bcd_deb}
		\FOR{$i = 1, \ldots,n_\text{iter} $} 
		\FOR{$k = 0,\ldots, L-1$}
		\STATE $\vct{z} \gets  \widetilde{\mtx{U}}^{(i)}_{k^c,k^c} \widetilde{\mtx{M}}_{k^c,k}$ \quad and \quad  $\gamma \gets \vct{z}^{\H} \widetilde{\mtx{M}}_{k^c,k}$
		\IF{$\gamma > 0$}
		\STATE $\widetilde{\mtx{U}}^{(i+1)}_{k^c,k} \text{and} \left({\widetilde{\mtx{U}}^{(i+1)}}_{k^c,k} \right)^{\H} \!\gets \!-\sqrt{\frac{1-\nu}{\gamma}} \vct{z}$\\  
		\ELSE
		\STATE $\widetilde{\mtx{U}}^{(i+1)}_{k^c,k} \text{and} \left({\widetilde{\mtx{U}}^{(i+1)}}_{k^c,k} \right)^{\H} \gets 0$
		\ENDIF
		\ENDFOR
		\ENDFOR \label{al_line:bcd_fin}
		\STATE $\widetilde{\vct{u}}$ = eigenvector associated to the largest eigenvalue of $\widetilde{\mtx{U}}^{(i+1)}$, and $\widetilde{\vct{u}} \gets \frac{\widetilde{\vct{u}}}{\widetilde{\vct{u}}[L-1]}$ \\
		\STATE $\xvbar \gets \Re\left(\Phivbar^{\H}\diag{(\vct{b})} \widetilde{\vct{u}}[:L-1] \right)$  and $\vct{x} = \begin{bmatrix} 
			\xvbar \\
			\xvnat
		\end{bmatrix}$
	\ENSURE{Reconstructed signal $\vct{x}$}
	\end{algorithmic}
\end{algorithm}

\section{Experiments}
\label{sec:experiments}

In this section we assess the potential of our methods for inpainting gaps in audio signals. Our code is available online.\footnote{\scriptsize \url{https://github.com/Louis-Bahrman/Inpainting-Fourier}}

\subsection{Experimental setting}
\label{sec:exp_setting}

\paragraph{Data}
We consider $100$ speech signals from the Librispeech dataset \cite{Panayotov2015}. Signals are sampled at $16$ kHz. For each signal, we extract a non-silent sub-signal of variable length $L$ at a random location, where we create a gap of $d$ samples.

\paragraph{Proposed methods}
For the AM algorithm, iterations stop when the the maximum loss variation over the previous $5$ iterations does not exceed $10^{-10}$, or when a maximum number of $1000$ iterations is reached. Preliminary experiments revealed no significant difference between various basic initialization schemes, e.g., using a random or zero phase in the Fourier domain, or using a random or zero missing signal in the time domain. We also derived an initialization scheme inspired by the spectral initialization method in~\cite{candes2015wirtinger}, but this did not yield any significant improvement either. Therefore, the results displayed hereafter use $\xvbar^{(0)}\leftarrow \vct{0}$.
The CR algorithm uses a fixed amount $n_\text{iter} = 10$ since the performance did not show further improvement beyond in our experiments. We simply set the barrier parameter $\nu$ to $0$ as in~\cite{Deleforge2017}. We also consider a \textit{combined} CR+AM algorithm. This technique consists in first estimating the signal with CR, and then using this estimate as an initialization for AM, with the same stopping criterion as above. Our methods are fed with the ground truth magnitudes ($\vct{b} = \abs{\mtx\Phi \vct{x}^{\natural}} $), except in the last experiment where noisy magnitudes are considered.

\paragraph{Baselines}
Let us note that typical phase retrieval methods~\cite{Gerchberg1972, waldspurger2015phase} are not appropriate comparison baselines since they are agnostic to the observed samples. Besides, more recent phase retrieval approaches proposed in audio~\cite{Vial2021,Masuyama2019} are tailored to time-frequency processing, and are therefore not suitable for the setup considered in this paper. On the other hand, traditional audio restoration techniques such as~\cite{janssen1986interpautoregressive} do not leverage magnitude knowledge, making comparisons unfair. Nevertheless, results obtained using the Fourier-sparsity-based method SPAIN~\cite{Mokry2019} are shown, as an indication of the potential of exploiting magnitude models beyond sparsity. The same parameters as in~\cite{Mokry2019} are used, except for the relaxation parameter $r$ which is set to $32$.

\paragraph{Metric}
To assess the quality of the reconstruction, we resort to the signal-to-error ratio (SER) expressed in dB:
\begin{equation}
    \text{SER}(\xvbar,\xvbar^{\natural} ) = 10 \log_{10}\left({\frac{\norm{\xvbar^{\natural}}^2}{\norm{\xvbar - \xvbar^{\natural}}^2}}\right),
 \label{eqn:ser}
\end{equation}
where $\vct{x}$ denotes the estimated signal (higher is better). Let us outline that the SER is only calculated over the set of missing samples $\bar{v}$ since $\vct{x}_{v} = \vct{x}^{\natural}_{v}$ elsewhere. Note that AM's initialization (replacing the missing signal with zeros) corresponds to an SER of $0$. \louis{Based on preliminary listening tests,} we consider that perfect reconstruction is achieved when the SER is greater than $20$~dB.

\subsection{Results}

\begin{figure}[t]
    \centering
    \includegraphics[width=0.8\linewidth]{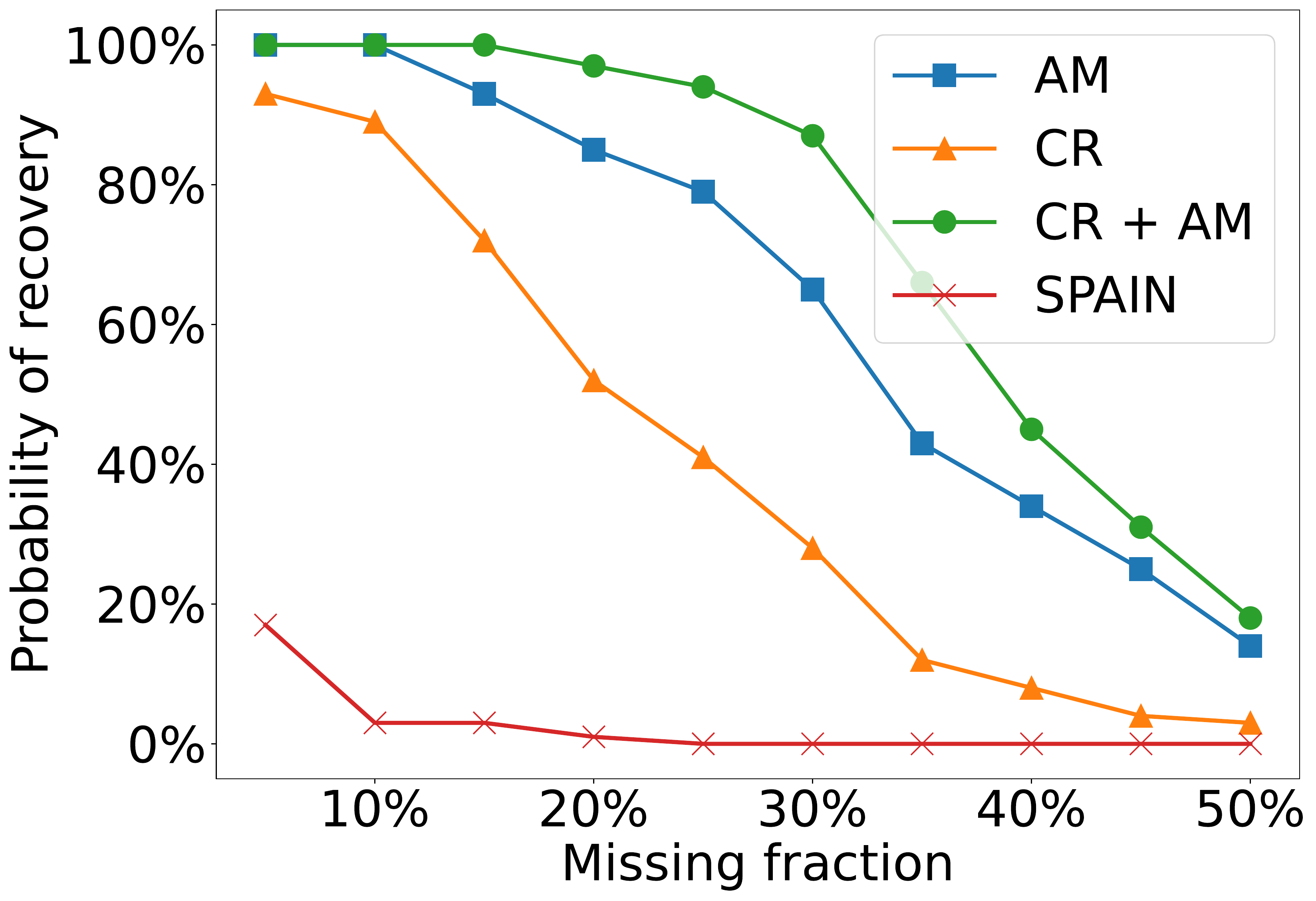}
    \vspace{-3mm}
    \caption{Influence of the missing fraction onto performance.}
   \label{fig:comparisonMethods}
\end{figure}

First, we compare the recovery rate of the three proposed methods with respect to the missing signal fraction. The results corresponding to $L=1024$ (64~ms) are displayed in Fig.~\ref{fig:comparisonMethods}. Our methods significantly outperform SPAIN in this setting, regardless of the missing fraction. This demonstrates the potential of exploiting Fourier magnitudes when these are available. CR+AM achieves the best results and consistently outperforms the other approaches. This is explained by the ability of CR to provide a solution that is more likely to converge to a global minimum than AM's basic initialization. While all the methods exhibit a performance drop when the missing fraction increases, CR+AM and AM still yield perfect reconstruction in at least $80$~\% of the cases when less than $30$~\% of the signal is missing. 
\louis{
The results of CR+AM have to be weighted against the higher memory and time complexity of CR, which processes matrices of size $L\times L$ instead of vectors of size $L$.
Accordingly, 
a general guideline might be to resort to the AM method when very few samples are missing ($5$ to $10$ \%), and to turn to CR+AM in all remaining scenarios, at a higher computational cost.
}

Then, Fig.~\ref{fig:influenceParameters} shows the recovery rate achieved by the best performing method CR+AM as a function of the signal length $L$ (which ranges from $128$ to $4096$ samples) and the fraction of missing samples $d/L$ (which ranges from $5$~\% to $50$~\%). We observe that perfect reconstruction is achieved in most cases when the fraction is lower than $33$~\% and ${L\le1024}$, suggesting that the approach performs near the theoretical optimum of~\cite{kreme2023inpainting} in this regime. While for short signals the method even yields satisfactory results under slightly larger missing fractions, its performance eventually drops for larger signals. This may be explained by the increased dimensionality of the problem which may increase the risk of getting trapped in local minima.

\begin{figure}[t]
    \centering
    \includegraphics[width=0.8\linewidth]{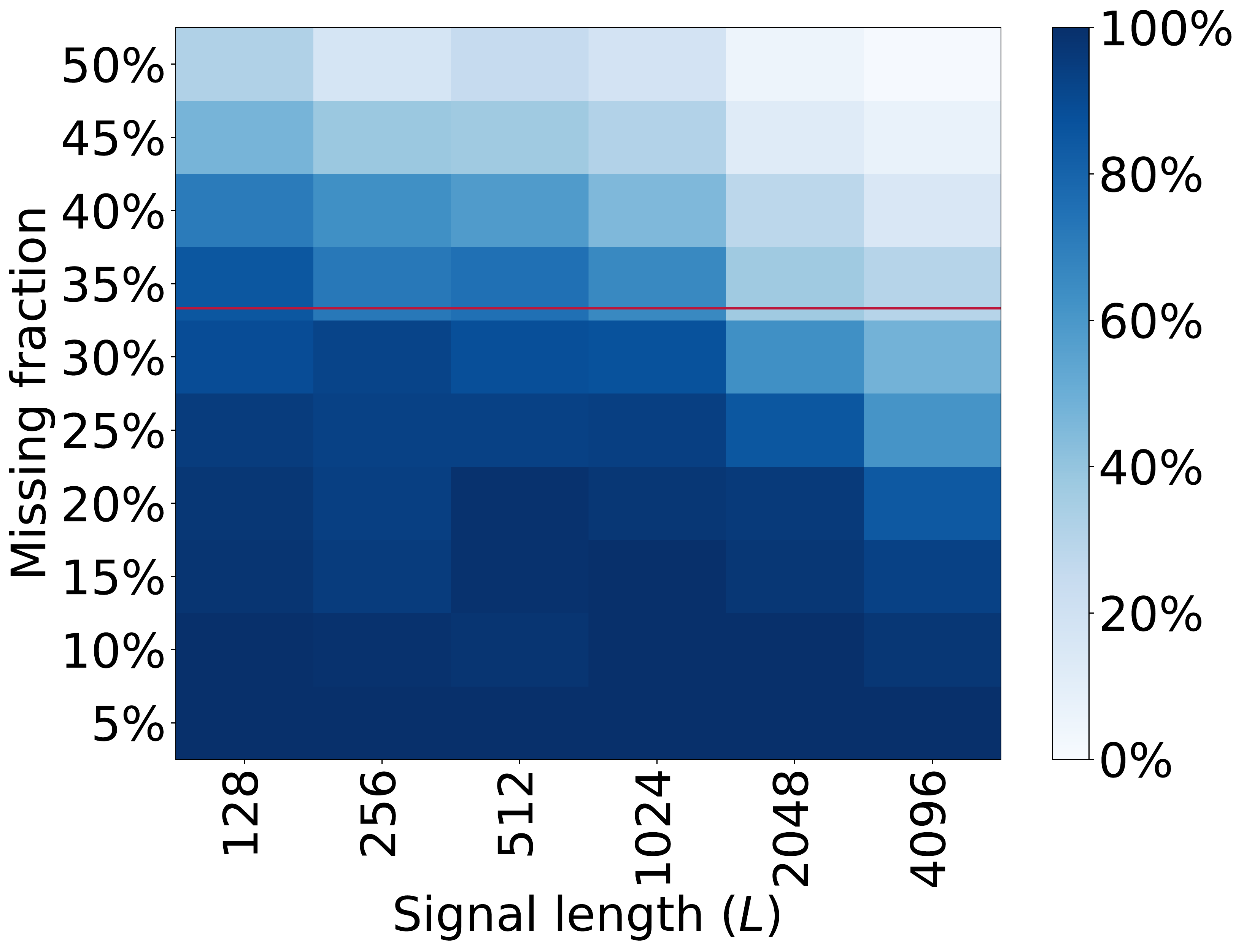}
    \vspace{-4mm}
    \caption{Probability of achieving perfect reconstruction for CR+AM, as a function of the signal length (in samples) and the missing signal fraction. The red line represents the $33$ \% theoretical limit under which perfect reconstruction is achievable for almost all signals~\cite{kreme2023inpainting}.}
    \label{fig:influenceParameters}
\end{figure}

\begin{figure}[t]
    \centering
    \includegraphics[width=0.8\linewidth]{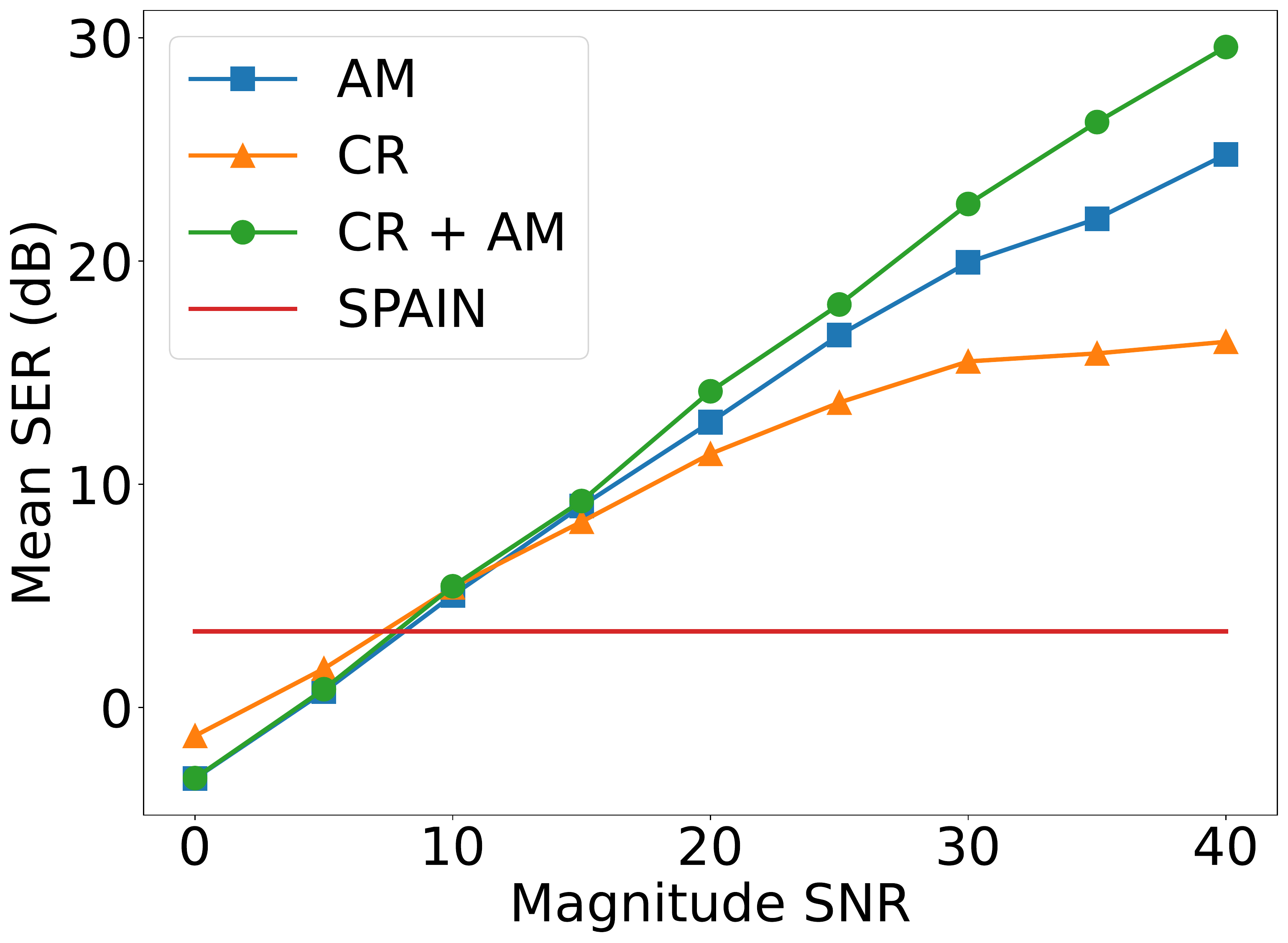}
    \vspace{-3mm}
    \caption{Influence of the magnitude noise onto performance.}
   \label{fig:noiseRobustness}
\end{figure}
Finally, let us evaluate the robustness of our methods to magnitude noise. Indeed, in the previous experiments we considered oracle magnitudes, but in practical applications these are estimated beforehand and therefore contaminated by estimation errors. To simulate such a scenario, we consider noisy magnitudes ${\vct{b} = \max{(0, \abs{\mtx\Phi \vct{x}^{\natural}} + \vct{n})}}$, where $\vct{n}$ is a white Gaussian noise whose variance is adjusted to fit a given magnitude signal-to-noise (SNR) value. The results are presented in Fig.~\ref{fig:noiseRobustness}, where the signal length is fixed at $L=1024$ and the fraction of missing signal at $25 \%$. Overall, we observe that all of our methods exhibit a similar linear decay in performance in the log-log plot when the SNR falls below $20$ dB, suggesting an encouraging robustness to errors. Interestingly, while CR performs worse than the other methods at higher SNRs, it still yields solutions that provide a better initialization than zero to AM.
Finally, we remark that SPAIN, which leverages Fourier-domain sparsity, outperforms our methods at low SNRs, while the opposite is true for SNRs above $8$~dB. This confirms that leveraging Fourier magnitude is promising, provided they have been accurately estimated beforehand. This motivates future research on magnitude estimation techniques.

\section{Conclusion}
\label{sec:conclusion}

We have investigated the problem of signal inpainting from observed Fourier magnitude measurements. After formulating the associated optimization problem, we proposed two estimation methods based on alternating minimization and convex relaxation. Experiments in the noiseless regime highlighted that combining the two methods yields near-exact reconstruction most of the time when less than $33\%$ of the signal is missing, closely following our theoretical limit. Further experiments demonstrated a relative robustness of the methods to magnitude errors. Future work will therefore focus on applying such techniques to more realistic settings where magnitudes are estimated, e.g., via light magnitude interpolation methods in the short-time Fourier domain or via data-driven models such as deep neural networks.

\bibliographystyle{IEEEtran}
\bibliography{biblio.bib}

\end{document}


\maketitle

\begin{abstract}
This document takes interest in signal inpainting from Fourier magnitudes. This task consists in reconstructing consecutive missing samples in a finite discrete 1D signal, while assuming the magnitudes of its Fourier transform are known. In this report, we theoretically show that for almost all signals of length $L$, this problem admits a unique solution if at most $(L-1)/3$ samples are missing.
\end{abstract}

\section{Introduction and problem setting}
Signal inpainting \cite{adler2012inpainting} is an inverse problem that consists in restoring signals degraded by sample loss. This problem typically arises as a result of degradation during signal transmission, digitization of physically degraded media, or degradation so heavy that the information about the samples can be considered lost~\cite{Rodbro2006,Chantas2018,Derebssa2018}. Let $\vct{x} \in \R^L$ be a signal. Let $\bar{v}\subset\{0,\dots,L-1\}$ denote a set of consecutive indices corresponding to missing samples in $\vct{x}$ and $v$ denote its complement, \textit{i.e.}, the set of indices corresponding to observed samples. We denote by $\xvbar\in\R^d$ the sub-signal of $\vct{x}$ restricted to missing samples and $\vct{x}_v\in\R^{L-d}$ the sub-signal restricted to observed samples. 
We denote by  $\vct{b} \in \R_+^L$ the magnitudes of the discrete Fourier transform (DFT) of $\vct{x}$, \textit{i.e.}, $\vct{b} = |\mtx{\Phi} \vct{x}|$, where $\mtx{\Phi} \in \C^{L \times L} $ is the DFT matrix. For a given observed signal $\vct{y}\in\R^{L-d}$ and Fourier magnitudes $\vct{b}$, the task of signal inpainting from Fourier magnitudes can then be stated as:
\begin{equation}
\label{eq:formulationInpainting}
    \mathrm{Find} \; \vct{u}\in\R^{d} \; \text{such that (s.t.)} \; \abs{\mtx{\Phi}  \vct{x}}  =  \vct{b} \; \text{with} \; \xvbar=\vct{u} \; \mathrm{and} \; \vct{x}_{v}=\vct{y}. 
\end{equation}
We focus on the situation where the given vector $\vct{b}$ corresponds to the true magnitudes of the Fourier transform of a completed signal $\vct{x}$. Hence, the existence of at least one solution of \eqref{eq:formulationInpainting} is guaranteed. In this document, we will show that when $d<(L-1)/3$, this solution is unique for \textit{almost all} signals $\vct{x}\in \R^L$. We use a dimension counting argument, similar in spirit to the one employed by \cite{jaganathan2017sparse} in the context of sparse phase retrieval. Specifically, we show that signals $\vct{x}$ for which more than one solution exists, referred to hereinafter as \textit{counter examples}, necessary lie on a manifold of $\R^L$ with strictly less than $L$ degrees of freedom. They hence form a set of measure zero, \textit{i.e.}, a null set.

\section{Almost uniqueness: statement and proof}
We will assume throughout this section that the missing samples are placed at the beginning of $\vct{x}$, allowing us to write $\vct{x}=[\vct{u} ; \vct{y}]$ where $[\cdot;\cdot]$ denotes vertical concatenation. This comes without loss of generality, because for any counter-example signal $\widetilde{\vct{x}}$ with consecutive missing samples placed anywhere in the signal, one can construct a counter example with samples placed at the beginning of the signal (and reciprocally) by a simple circular shift of $\widetilde{\vct{x}}$. Indeed, a circular shift does not affect DFT magnitudes. We prove the following theorem:

\begin{theorem}
Let $\mathcal{E}=\bigl\{\vct{x}=[\vct{u};\vct{y}] \in \R^L \;|\; \exists \vct{v} \in\R^d, \vct{v}\ne \vct{u}, \; \textrm{s.t.} \; \abs{ \mtx{\Phi}  [\vct{u};\vct{y}] } = \abs{ \mtx{\Phi}  [\vct{v};\vct{y}] } \bigr\}$. For $d<(L-1)/3$, $\mathcal{E}$ is a manifold of $\R^L$ with strictly less than $L$ degrees of freedom.
\end{theorem}
In other words, the set of counter examples to the unicity of \eqref{eq:formulationInpainting} has strictly less than $L$ degrees of freedom in $\R^d$, and is hence of measure zero.

\begin{proof} We denote by $\mathcal{E}'$ the set of triplets $(\vct{y},\vct{u},\vct{v})\in\R^{L-d}\times\R^d\times\R^d$ forming a counter example, namely:
\begin{equation}
    \mathcal{E}'=\bigl\{(\vct{y},\vct{u},\vct{v})\in\R^{L-d}\times\R^d\times\R^d \;|\; \vct{v}\ne \vct{u}, \; \abs{ \mtx{\Phi}  [\vct{u};\vct{y}] } = \abs{ \mtx{\Phi}  [\vct{v};\vct{y}] } \bigr\}.
\end{equation}
We will show that this manifold of $\R^{L-d}\times\R^d\times\R^d$ has strictly less that $L$ degrees of freedom; This implies that its projection $\mathcal{E}$ on $\R^{L-d}\times\R^d\equiv\R^L$ also has strictly less that $L$ degrees of freedom. We first prove the following:

\begin{lemma}
\label{lem:bijection}
There is a linear bijection between $\mathcal{E}'$ and the following set:
\begin{equation}
    \label{eq:az_set}
    \mathcal{E}'' = \left\{(\vct{a},\vct{w})\in\R^{d}\times\R^L \;|\; \vct{a}\ne \vct{0}_d,\; \mathcal{R}\left(\overline{\mtx{\Phi}[\vct{a};\vct{0}_{L-d}]}\odot \mtx{\Phi}\vct{w}\right) = \vct{0}_L \right\}
\end{equation}
where, $\mathcal{R}(\cdot)$ denotes the real part of a vector and $\odot$ denotes element-wise product.
\end{lemma}
\begin{proof}
We horizontally split the DFT matrix as $\mtx{\Phi}=[\mtx{\Phi}^{(1)},\mtx{\Phi}^{(2)}]$ where $\mtx{\Phi}^{(1)}\in\R^{L\times d}$ and $\mtx{\Phi}^{(2)}\in\R^{L\times L-d}$ We have the following chain of equivalences (where $\vct{v}\ne \vct{u}$ is kept implicit):
\begin{align}
    &(\vct{y},\vct{u},\vct{v})\in\mathcal{E}' \\
    \Leftrightarrow & \;
    \abs{ \mtx{\Phi}  [\vct{u};\vct{y}] }^2 = \abs{ \mtx{\Phi}  [\vct{v};\vct{y}]}^2 \\
    \Leftrightarrow & \;
    \abs{\mtx{\Phi}^{(1)}\vct{u}
    + \mtx{\Phi}^{(2)}\vct{y}}^2
    = \abs{\mtx{\Phi}^{(1)}\vct{v}
    + \mtx{\Phi}^{(2)}\vct{y}}^2 \\
    \Leftrightarrow & \;
    \abs{\mtx{\Phi}^{(1)}\vct{u}}^2 
    + 2\mathcal{R}\left(\overline{\mtx{\Phi}^{(1)}\vct{u}}\odot\mtx{\Phi}^{(2)}\vct{y}\right)
    + \abs{\mtx{\Phi}^{(2)}\vct{y}}^2
    = \abs{\mtx{\Phi}^{(1)}\vct{v}}^2 
    + 2\mathcal{R}\left(\overline{\mtx{\Phi}^{(1)}\vct{v}}\odot\mtx{\Phi}^{(2)}\vct{y}\right)
    + \abs{\mtx{\Phi}^{(2)}\vct{y}}^2 \\
    \Leftrightarrow & \;
    \abs{\mtx{\Phi}^{(1)}\vct{u}}^2
    - \abs{\mtx{\Phi}^{(1)}\vct{v}}^2
    + 2\mathcal{R}\left(\overline{\mtx{\Phi}^{(1)}(\vct{u}-\vct{v})}\odot\mtx{\Phi}^{(2)}\vct{y}\right)
    = \vct{0}_L \\
    \Leftrightarrow & \;
    \mathcal{R}\left(
         (\overline{\mtx{\Phi}^{(1)}\vct{u}} - \overline{\mtx{\Phi}^{(1)}\vct{v}}) \odot(\mtx{\Phi}^{(1)}\vct{u} + \mtx{\Phi}^{(1)}\vct{v})
          \right)
    + 2\mathcal{R}\left(\overline{\mtx{\Phi}^{(1)}(\vct{u}-\vct{v})}\odot\mtx{\Phi}^{(2)}\vct{y}\right)
    = \vct{0}_L  \\
    \Leftrightarrow & \;
    \mathcal{R}\left(\overline{\mtx{\Phi}^{(1)}(\vct{u} - \vct{v})}\odot(\mtx{\Phi}^{(1)}\vct{u} + \mtx{\Phi}^{(1)}\vct{v} + 2\mtx{\Phi}^{(2)}\vct{y}\right)
    = \vct{0}_L  \\
    \Leftrightarrow & \;
    \mathcal{R}\left(\overline{\mtx{\Phi}[\vct{u} - \vct{v};\vct{0}_{L-d}]}\odot(\mtx{\Phi}[\vct{u} + \vct{v}; 2\vct{y}]\right)
    = \vct{0}_L  \\
    \Leftrightarrow & \;
    \mathcal{R}\left( \overline{\mtx{\Phi}[\vct{a};\vct{0}_{L-d}]} \odot \mtx{\Phi}\vct{w} \right) = \vct{0}_L
    \;\textrm{where}\; \vct{a} = \vct{u}-\vct{v} \ne \vct{0}_d \; \textrm{and} \; \vct{w} = [\vct{u} + \vct{v}; 2\vct{y}]\in\R^L.
\end{align}
Since the transformation from $(\vct{y},\vct{u},\vct{v})$ to $(\vct{a},\vct{w})$ is linear and bijective, this concludes the proof.
\end{proof}
Based on Lemma~\ref{lem:bijection}, it is sufficient to show that $\mathcal{E}''$ has strictly less than $L$ degrees of freedom. Since the non-zero signal $\vct{a}\in\R^d$ in \eqref{eq:az_set} can be chosen arbitrarily ($d$ degrees of freedom), it remains to show that for a fixed $\vct{a}\ne\vct{0}_d$, the set of $\vct{w}\in\R^L$ such that $(\vct{a},\vct{w})\in\mathcal{E}''$ has strictly less than $L-d$ degrees of freedom. For conciseness, we will only treat here the case where $L$ is even, as the odd case only requires minor adjustments.

Let $\hat{\vct{a}}=\mtx{\Phi}[\vct{a};\vct{0}_{L-d}]$ and  $\hat{\vct{w}}=\mtx{\Phi}\vct{w}$ be the DFTs of $[\vct{a};\vct{0}_{L-d}]$ and $\vct{w}$, indexed by the $L$ discrete frequency numbers $f\in\{-L/2+1,\dots,L/2\}$.
Since the signals $\vct{a}$ and $\vct{w}$ are real-valued, their DFTs are fully determined by their values at non-negative frequencies, two of which are real (at $f=0$ and $f=L/2$), the rest being complex.
For every frequency number $f\in\{1,\cdots,L/2-1\}$ such that $\hat{a}(f)\ne0$, the constraint $\mathcal{R}(\overline{\hat{\vct{a}}} \odot \hat{\vct{w}}) = \vct{0}_L$ fixes the phase of $\hat{w}(f)$ up to $\pm \pi/2$, reducing the degrees of freedom of $\vct{w}$ by 1 (from a total of $L$). Let us now count for how many distinct $f\in\{1,\cdots,L/2-1\}$ we can have $\hat{a}(f)\ne0$. The z-transform of $[\vct{a};\vct{0}_{L-d}]$ is a polynomial of degree at most $d-1$. Hence, this polynomial admits at most $d-1$ roots, and since $\vct{a}$ is real-valued, these roots are either real or come in conjugate pairs. This implies that $\hat{a}(f)$ can be 0 for at most $\lfloor(d-1)/2\rfloor$ distinct $f$ in $\{1,\cdots,L/2-1\}$. Hence, the constraint$\mathcal{R}(\overline{\hat{\vct{a}}} \odot \hat{\vct{w}}) = \vct{0}_L$ enforces at least $L/2-1-\lfloor(d-1)/2\rfloor$ phase constraints on $\hat{\vct{w}}$. Subtracting these constraints from $L$, we get that $\vct{w}$ has at most $P = L/2+1+\lfloor(d-1)/2\rfloor$ degrees of freedom.  By hypothesis, $d<L/3-1$, which implies $P<L-d$ and concludes the proof.
\end{proof}

\section{Conclusion}
We have conducted a theoretical study on the solutions to the problem of signal inpainting from Fourier magnitudes.
We have shown that if the number of missing samples $d$ is strictly less than $(L-1)/3$, where $L$ is the total signal length, then almost all signals containing a subset of $d$ consecutive missing values are uniquely determined from the magnitudes of their Fourier transform.

\bibliographystyle{plain}
\bibliography{biblio.bib}